\documentclass[submission,copyright,creativecommons]{eptcs}
\usepackage{xspace}
\usepackage{amsmath,amssymb,amsthm}
\usepackage{latexsym}
\usepackage{txfonts}
\usepackage{enumitem}
\usepackage{tikz}
\usetikzlibrary{arrows,automata,decorations.pathreplacing,shapes,shapes.geometric,calc,fit,positioning}

\theoremstyle{plain}
\newtheorem{theorem}{Theorem}
\newtheorem{proposition}[theorem]{Proposition}
\newtheorem{corollary}[theorem]{Corollary}
\newtheorem{lemma}[theorem]{Lemma}

\theoremstyle{definition}
\newtheorem{definition}[theorem]{Definition}
\newtheorem{example}[theorem]{Example}

\newcommand{\spoiler}{\textsc{Spoiler}\xspace}
\newcommand{\duplicator}{\textsc{Duplicator}\xspace}
\newcommand{\Nat}{\mathbb{N}}

\def\newarrow#1{\mathop{{\hbox{\setbox0=\hbox{$\scriptstyle{#1\quad}$}{$%
\mathrel{\mathop{\setbox1=\hbox to
\wd0{\rightarrowfill}\ht1=3pt\dp1=-2pt\box1}\limits^{#1}}%
$}}}}}

\renewcommand{\epsilon}{\varepsilon}

\newcommand{\Transition}[3]{\ensuremath{#1 \newarrow{#2} #3}}

\newcommand{\simugame}[3]{\ensuremath{\mathcal{G}^{#1}(#2,#3)}}
\newcommand{\simu}[1]{\ensuremath{\mathbin{\sqsubseteq^{#1}}}}
\newcommand{\notsimu}[1]{\ensuremath{\mathbin{\nsqsubseteq^{#1}}}}

\newcommand{\project}[2]{\ensuremath{#1{\downarrow}_{#2}}}

\title{Two-Buffer Simulation Games}
\author{Milka Hutagalung$^1$ \qquad Norbert Hundeshagen$^1$ \qquad Dietrich Kuske$^2$ \\ Martin Lange$^1$  \qquad Etienne Lozes$^3$  
\institute{$^1$ School of Electrical Engineering and Computer Science \\
University of Kassel, Germany \\
$^2$ Technische Universit\"at Ilmenau, Germany \\
$^3$ LSV, ENS Cachan, France}
}
\def\titlerunning{Two-Buffer Simulations}
\def\authorrunning{M.~Hutagalung, N.~Hundeshagen, D.~Kuske, M.~Lange $\&$  E.~Lozes}
\def\copyrightholders{M.~Hutagalung, N.~Hundeshagen,\\ D.~Kuske, M.~Lange, E.~Lozes}
\begin{document}
\maketitle

\begin{abstract}
We consider simulation games played between Spoiler and Duplicator on two B\"uchi automata 
in which the choices made by Spoiler can be buffered by Duplicator in two different buffers before she 
executes them on her structure.
Previous work on such games using a single buffer has shown that they are useful to approximate 
language inclusion problems. We study the decidability and complexity and show that games with 
two buffers can be used to approximate corresponding problems on finite transducers, i.e.\ the
inclusion problem for rational relations over infinite words.
\end{abstract}


\section{Introduction}

Simulation is a pre-order between labelled transition systems $\mathcal{T}$ and $\mathcal{T}'$ that formalises the idea
that ``$\mathcal{T}'$ can do everything that $\mathcal{T}$ can''. Formally, it relates the states of the two transition
systems such that each state $t'$ in $\mathcal{T}'$ that is connected to some $t$ in $\mathcal{T}$ can mimic the immediate
behaviour of $t$, i.e.\ it carries the same label, and whenever $t$ has a successor then $t'$ has a matching one, too. 

Simulation relations have become popular in the area of automata theory because they can be used to efficiently approximate
language inclusion problems for automata on finite or infinite words and trees and to minimise such automata 
\cite{lncs575*255,EtessamiWS01,journals/tcs/FritzW05,conf/tacas/AbdullaBHKV08}. Simulation relationship is usually
computable in polynomial time whereas language inclusion problems are PSPACE-complete for typical (finite, B\"uchi, parity,
etc.) automata on words and EXPTIME-complete for such automata on trees. Approximation is to be understood in this context as
follows: a positive instance of simulation yields a positive instance of language inclusion but not necessarily vice-versa.

Games played between two players -- usually called \spoiler and \duplicator -- on the state 
spaces of two automata yield a very intuitive characterisation which can be used to reason about simulations.
It is very easy to construct examples of pairs of automata such that language inclusion holds but simulation does not, i.e.\ 
the game makes \duplicator too weak. This has led to the study of several extensions of simulation relations and games with
the aim of making \duplicator stronger or \spoiler weaker whilst retaining a better complexity than language inclusion. Two examples 
in this context are \emph{multi-pebble simulation} \cite{conf/concur/Etessami02} and \emph{multi-letter simulation} \cite{HutagalungLL:LATA13,MayrC13}. 
\begin{itemize}
\item In \emph{multi-pebble simulation} \cite{conf/concur/Etessami02}, \duplicator controls several pebbles and can therefore 
     act in foresight of several of \spoiler's later moves. This kind of simulation is still computable in polynomial time 
     for any fixed number of pebbles. It forms a hierarchy of more and more refined simulation relations that approximate  
     language inclusion.
\item In \emph{multi-letter simulation} \cite{HutagalungLL:LATA13,MayrC13}, \spoiler is forced to reveal more than one transition
     of the run that he constructs in the limit. Again, these relations are computable in polynomial time for any fixed look-ahead, that is the
     number of symbols that \spoiler advances in the construction of his run ahead of \duplicator's. In fact, all are 
     computable in time linear in the size of each underlying automaton -- unlike multi-pebble simulations --, and polynomial
     of a higher degree only in the size of the underlying alphabet. They also form a hierarchy between ordinary simulation and 
     language inclusion.
\end{itemize}
In multi-letter simulations \spoiler is forced to reveal more than one transition in each round. This look-ahead can be realised using a bounded 
FIFO buffer which is filled by \spoiler and emptied by \duplicator. This has led to another extension of simulation.
\begin{itemize}
\item In \emph{buffered simulations} \cite{DBLP:journals/corr/HutagalungLL14}, the letters chosen by \spoiler get stored in
     an unbounded FIFO buffer, and \duplicator consumes them to form her moves. These games show a limit of what is 
     possible in terms of approximating language inclusion: buffered simulation is in general EXPTIME-complete \cite{DBLP:journals/corr/HutagalungLL14}, 
     i.e.\ even
     harder than language inclusion, yet there are pairs of B\"uchi automata on which \spoiler wins the buffered simulation
     game even though language inclusion holds. This is only true for $\omega$-languages, though. On finite words, buffered
     simulation captures language inclusion because \duplicator can simply wait for \spoiler to produce an entire word before
     she makes any moves. 

     This may raise the question of why buffered simulation is EXPTIME-complete when it captures a PSPACE-complete problem in
     this case. There is in fact a natural restriction of buffered simulation which is only PSPACE-complete \cite{DBLP:journals/corr/HutagalungLL14};
     it requires \duplicator to always flush the entire buffer when she moves. Note that simulation games on automata on finite
     words can be played with this restriction since they degenerate to games with no proper alternation: first \spoiler produces
     a word, then \duplicator consumes it entirely from the buffer.
\end{itemize}
In this paper we consider a natural extension of simulation games
played with FIFO buffers, namely \emph{two-buffer} simulations. Again,
\spoiler and \duplicator each move a pebble along the transitions of
two B\"uchi automata, forming runs, and \duplicator wins a play if her
run is accepting or \spoiler's is not. The letters chosen by \spoiler
get put into one of two buffers from which \duplicator takes letters
to form her run.  Note that two-buffer games do \emph{not} approximate
language inclusion on B\"uchi automata because the availability of
more than one buffer introduces a commutativity property between
alphabet symbols.  Since this commutativity is based on a partition of
the alphabet, it matches the commutativity in the direct product of
two free monoids. Consequently, as an application, we obtain
polynomial-time computable approximation procedures for the otherwise
undecidable inclusion problem between rational relations over infinite
words.

The paper is organised as follows. Section~\ref{sec:prel} recalls necessary definitions etc.\ on B\"uchi automata, simulation
games and buffered simulation games. In Section~\ref{sec:multibuf} we introduce and study two-buffer simulations. 
We show that they are potentially interesting from a point of efficiency: they can be decided in polynomial time for any fixed number
of bounded buffers and any fixed buffer capacity. We also examine when undecidability occurs: it is not surprising that using
two unbounded buffers leads to undecidability since problems involving two unbounded FIFO buffers are typically undecidable,
for instance the reachability problem for a network of finite state machines communicating through unbounded perfect FIFO 
channels is undecidable for various set-ups~\cite{BrandZ83,ChambartS08}. 
In Section~\ref{sec:transduce} we present our application of
two-buffer simulation games which features the aforementioned partial
commutativity to the inclusion problem between rational relations over
infinite words.

\section{Preliminaries}
\label{sec:prel}

\subsubsection*{Automata.} A \emph{B\"uchi automaton} is a tuple
$\mathcal{A} = (Q,\Sigma,q_I,\delta,F)$ where $Q$ is a finite set of
states with a designated starting state $q_I \in Q$, $\Sigma$ is the
underlying alphabet, $\delta \subseteq Q \times \Sigma \times Q$ is
the transition relation and $F \subseteq Q$ is a designated set of
accepting states.

The language $L(\mathcal{A})$ of a B\"uchi automaton is, as usual, the
set of all infinite words for which there is a run starting in $q_I$
that visits some states in $F$ infinitely often. It is known that
single-buffer simulation games, as recalled below, can be used to
approximate B\"uchi language inclusion problems
\cite{DBLP:journals/corr/HutagalungLL14}. However, the simulation
games introduced in Section~\ref{sec:multibuf} lead away from the
problem of language inclusion between B\"uchi automata. Hence, we are
not particularly concerned with B\"uchi automata as acceptors of
$\omega$-languages. Instead it suffices at this point to simply see
them as directed graphs with labelled edges, a designated starting
state and some marked states that will be used to define winning
conditions in games played on these graphs.  We will therefore rather
write $\Transition{q}{a}{p}$ than $(q,a,p) \in \delta$ when the
underlying transition relation is clear from the context. We will also
simply speak of \emph{automata} rather than \emph{B\"uchi automata} to
take away the focus from plain language inclusion problems.

\subsubsection*{Buffered simulation.} The \emph{buffered simulation} game \cite{DBLP:journals/corr/HutagalungLL14} is played between \spoiler and \duplicator
on two automata $\mathcal{A} = (Q^{\mathcal{A}},\Sigma,q_I,\delta^{\mathcal{A}},$ $F^{\mathcal{B}})$ and 
$\mathcal{B} = (Q^{\mathcal{B}},\Sigma,p_I,\delta^{\mathcal{B}},F^{\mathcal{B}})$ on configurations of the form $(q,\beta,p)$ with $q \in Q^{\mathcal{A}}$,
$p \in Q^{\mathcal{B}}$ and $\beta \in \Sigma^*$. This additional component acts like a buffer through which the alphabet symbols chosen by \spoiler get
channelled before \duplicator can use them. Such a buffer has a capacity $k \in \Nat \cup \{\omega\}$. If $k \in \Nat$ then the buffer is called
\emph{bounded}, for $k= \omega$ it is called \emph{unbounded}.

Any play of that game starts in the configuration $(q_I,\epsilon,p_I)$. Any round that has reached a configuration $(q,\beta,p)$ proceeds as
follows.
\begin{enumerate}
\item \spoiler chooses a $q'\in Q^{\mathcal{A}}$ and $a\in\Sigma$ such
  that $\Transition{q}{a}{q'}$. Let $\beta' := a\beta$.
\item \duplicator can either skip her turn in which case the play
  proceeds in the configuration $(q',\beta',p)$. Or we have
  $\beta' = \beta'' b$ for some $b \in \Sigma$. In this case she can choose a
  $p'\in Q^{\mathcal B}$ such that $\Transition{p}{b}{p'}$, and the
  play proceeds with $(q',\beta'',p')$.
\end{enumerate}
A play $(q_0,\beta_0,p_0),(q_1,\beta_1,p_1),\ldots$ is won by \duplicator iff
\begin{itemize}
\item  $|\beta_i| \leq k$ for all $i \in \Nat$ , i.e.\ the buffer never exceeds its capacity, and
\item either 
\begin{itemize}
\item there are only finitely many $i$ such that $q_i \in F^\mathcal{A}$, or
\item \duplicator moves infinitely often, and there are infinitely many $i$ such that $p_i \in F^{\mathcal{B}}$.

\end{itemize}
\end{itemize} 
Otherwise \spoiler wins. We write $\mathcal{A} \simu{k} \mathcal{B}$ to state that duplicator has a winning strategy for the buffered simulation game on
$\mathcal{A}$ and $\mathcal{B}$ with buffer capacity $k$.

Note that according to these winning conditions, the buffer capacity is only checked after \duplicator's moves. Hence, it is possible to play
on a buffer of capacity $k=0$; this just means that \duplicator has to consume the alphabet symbol chosen by \spoiler immediately.
We also remark that in this definition of buffered simulation game, \duplicator only ever consumes a single alphabet symbol in each of her
moves. One could equally allow her to consume several symbols from the buffer. This has no immediate advantage for her; if she has a winning
strategy then she also has one in which she only ever makes moves on one symbol. The reason for this is the simple fact that the buffer 
gives her a look-ahead on the choices made by \spoiler in an ordinary simulation game, and winning is monotone in the amount of information
available about the opponent's future moves.

For two automata $\mathcal A$ and $\mathcal B$, we write
$\mathcal A\simu{}\mathcal B$ if \duplicator has a winning strategy in
the ordinary fair simulation game~\cite{EtessamiWS01}. 

\begin{proposition}[\cite{HutagalungLL:LATA13,DBLP:journals/corr/HutagalungLL14}]
\label{prop:hierarchy}
We have
\begin{enumerate}[label=\alph*)]
\item \label{prop:hierarchy:ordsimu} $\simu{} = \simu{0}$, 
\item $\simu{k} \subsetneq \simu{k+1}$ for any $k \in \Nat$,
\item $\bigcup_{k\geq 0}\simu{k} \subsetneq \simu{\omega}$,
\item\label{prop:hierarchy:simimpliesli}
 $\mathcal{A} \simu{\omega} \mathcal{B}$ implies $L(\mathcal{A}) \subseteq L(\mathcal{B})$ for any B\"uchi automata $\mathcal{A}$ and
      $\mathcal{B}$. The converse does not hold in general.
\end{enumerate} 
\end{proposition}

\section{Two-Buffer Simulations}
\label{sec:multibuf}

\subsection{Simulation Games Using Two Buffers}
We fix two automata $\mathcal{A} = (Q^{\mathcal{A}},\Sigma,q_I,\delta^{\mathcal{A}},$ $F^{\mathcal{B}})$ and 
$\mathcal{B} = (Q^{\mathcal{B}},\Sigma,p_I,\delta^{\mathcal{B}},F^{\mathcal{B}})$ and develop the theory of two-buffer simulation games with respect
to these two automata. 

We assume a mapping $\sigma: \Sigma \to \{1,2\}$ which assigns a buffer index to each input symbol. Hence, the buffer that a symbol is put into
is determined by the symbol itself, not by one of the players.

\begin{definition} \rm Let $(k_1,k_2)\in(\Nat\cup\{\omega\})^2$ be a
  pair of buffer capacities, i.e.\ buffers can have bounded or
  unbounded
  capacity.
  The \emph{two-buffer simulation game with capacities $(k_1,k_2)$},
  denoted $\simugame{k_1,k_2}{\mathcal{A}}{\mathcal{B}}$, is played
  between players \spoiler and \duplicator on $\mathcal{A}$ and
  $\mathcal{B}$ using two FIFO-buffers, of capacities $k_1, k_2$ and
  initially empty, and two tokens which are initially placed on $q_I$
  and $p_I$. A configuration is denoted by $(q,\beta_1,\beta_2,p)$,
  consisting of the two states which currently carry the tokens on the
  left and right and the contents of the two buffers written as finite
  words over $\Sigma$ in between.  The initial configuration is
  $(q_I,\epsilon,\epsilon,p_I)$. A round consists of a move by player
  \spoiler followed by a move by \duplicator.  In a configuration of
  the form $(q,\beta_1,\beta_2,p)$,
\begin{enumerate}
\item \spoiler selects a $q'\in Q^{\mathcal A}$ and an $a \in \Sigma$ such that $\Transition{q}{a}{q'}$. Let $i = \sigma(a)$ and update the buffers as follows:
      $\beta_i := a\beta_i$, i.e.\ $a$ gets appended to the $i$-th buffer. 
\item \duplicator repeats the following step $k$ times for some $k$ with $0 \le k \le |\beta_1|+|\beta_2|$.
            \begin{itemize}
\item She selects a $j \in \{1,2\}$ such that $\beta_j = \gamma b$ for some $b \in \Sigma$, as well as a $p'\in Q^{\mathcal B}$ with $\Transition{p}{b}{p'}$.
              Then she updates the buffers as follows: $\beta_j := \gamma$, i.e.\ $b$ gets removed from the $j$-th buffer. She proceeds with 
              $(q',\beta_1,\beta_2,p')$, i.e.\ with the new state $p'$ instead of $p$.
\end{itemize}
\end{enumerate}
A play is a sequence $(q_0,\beta^0_1,\beta^0_2,p_0),(q_1\beta^1_1,\beta^1_2,p_1),\ldots$ of configurations. 
It is won by \duplicator iff
\begin{itemize}
\item $|\beta^i_j| \leq k_j$ for all $i \in \Nat$ and $j\in\{1,2\}$,
  i.e.\ no buffer ever exceeds its capacity, and
\item either 
\begin{itemize}
\item there are only finitely many $i$ such that $q_i \in F^\mathcal{A}$, or
\item every alphabet symbol that gets put into one of the buffers also gets removed from it eventually, and there are infinitely many $i$ such 
      that $p_i \in F^{\mathcal{B}}$.
\end{itemize}
\end{itemize} 
Otherwise \spoiler wins.
\end{definition}

\begin{definition}\rm
  Given $k_1,k_2$ as above, we say that $\mathcal{B}$
  \emph{$(k_1,k_2)$-simulates} $\mathcal{A}$, written
  $\mathcal{A} \simu{k_1,k_2}{} \mathcal{B}$ if player \duplicator has
  a winning strategy for the game
  $\simugame{k_1,k_2}{\mathcal{A}}{\mathcal{B}}$.
\end{definition}

\begin{example}
\label{ex:multibuff}
Consider the following two NBA $\mathcal{A}$ (left) and $\mathcal{B}$
(right) over the alphabet $\Sigma = \{ a,b,c,d\}$ with $\sigma(a)=1$, $\sigma(b)=\sigma(c)=\sigma(d)=2$. 
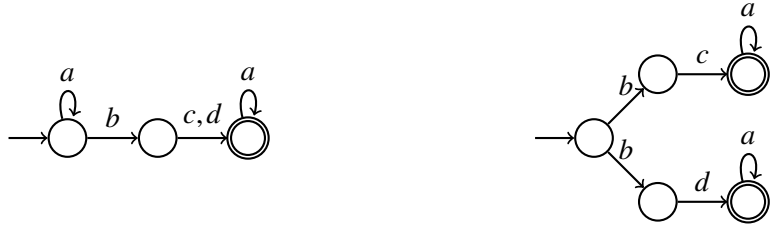
\begin{center}
\begin{tikzpicture}[initial text={}, node distance=12mm, every state/.style={minimum size=5mm}, thick]
  \node[state,initial]   (qa1)       {};
  \node[state]           (qa2) [right of=qa1]       {};
  \node[state,accepting] (qa3) [right of=qa2]       {};
  
  \path[->] (qa1) edge              node [above]      {$b$}        (qa2)
            (qa1) edge [loop above] node [above]      {$a$}        (qa2)
            (qa2) edge              node [above]      {$c,d$}        (qa3)
            (qa3) edge [loop above] node              {$a$}        (qa3);
            
  \node[state,initial]   (qb0) [right of=qa1, node distance = 7cm] {};
  \node[state]           (qb1) [above right of=qb0]                      {};
  \node[state]           (qb2) [below right of=qb0]                      {};
  \node[state,accepting] (qb3) [right of=qb1]                      {};
  \node[state,accepting] (qb4) [right of=qb2]                      {};
  
  \path[->] (qb0) edge              node [above]      {$b$}        (qb1)
            (qb0) edge              node [above]      {$b$}        (qb2)
            (qb1) edge              node [above]      {$c$}        (qb3)
            (qb2) edge              node [above]      {$d$}        (qb4)
            (qb4) edge [loop above] node              {$a$}        (qb4)
            (qb3) edge [loop above] node              {$a$}        (qb3);   

\end{tikzpicture}
\end{center}
We have $\mathcal{A} \simu{\omega,1} \mathcal{B}$, since \duplicator
has a winning strategy: skipping her turn until \spoiler reads $c$ or
$d$.  \duplicator then consumes the entire content of the two buffers
and reaches the accepting loop.
\end{example}

We describe a second example in order to explain a possibly apparent difference in the definitions of the single- and the two-buffer games: in 
single-buffer games we let \duplicator consume at most one symbol per round only because allowing her to consume more than one does not make her 
stronger. There are two-buffer games, though, which she can only win when she is allowed to consume more than one symbol at a time. 
In such games at least one buffer must be bounded; if both buffers are unbounded then \duplicator can defer her moves for any finite number of
times at any point and can never be forced by \spoiler to consume a letter in order to not exceed a buffer capacity. Now if there are two 
buffers then in order to consume an element from one that is close to overflow she may have to first consume one from the other buffer, hence
consume more than one letter in a round. 

\begin{example}
Consider the following two NBA $\mathcal{A}$ (left) and $\mathcal{B}$
(right) over the alphabet $\Sigma = \{ a,b\}$ with $\sigma(a)=1$, $\sigma(b)=2$, and the play $\simugame{\omega,1}{\mathcal{A}}{\mathcal{B}}$. 
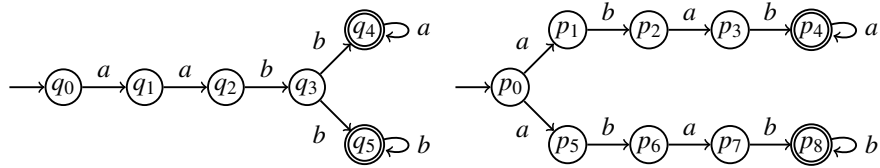
\begin{center}
\scalebox{0.9}{\begin{tikzpicture}[initial text={}, node distance=12mm, every state/.style={minimum size=4mm,inner sep=1pt}, thick]
  \node[state,initial]   (qa0)                      {$q_0$};
  \node[state]           (qa1) [right of=qa0]       {$q_1$};
  \node[state]           (qa2) [right of=qa1]       {$q_2$};
  \node[state]           (qa3) [right of=qa2]       {$q_3$};
  \node[state,accepting] (qa4) [above right of=qa3] {$q_4$};
  \node[state,accepting] (qa5) [below right of=qa3] {$q_5$};
  
  \path[->] (qa0) edge              node [above]      {$a$} (qa1)
            (qa1) edge              node [above]      {$a$} (qa2)
            (qa2) edge              node [above]      {$b$} (qa3)
            (qa3) edge              node [above left] {$b$} (qa4)
                  edge              node [below left] {$b$} (qa5)
            (qa4) edge [loop right] node [right]      {$a$} ()
            (qa5) edge [loop right] node [right]      {$b$} ();
            
  \node[state,initial]   (qb0) [right of=qa3, node distance = 3cm] {$p_0$};
  \node[state]           (qb1) [above right of=qb0]                {$p_1$};
  \node[state]           (qb2) [right of=qb1]                      {$p_2$};
  \node[state]           (qb3) [right of=qb2]                      {$p_3$};
  \node[state,accepting] (qb4) [right of=qb3]                      {$p_4$};
  \node[state]           (qb5) [below right of=qb0]                {$p_5$};
  \node[state]           (qb6) [right of=qb5]                      {$p_6$};
  \node[state]           (qb7) [right of=qb6]                      {$p_7$};
  \node[state,accepting] (qb8) [right of=qb7]                      {$p_8$};
  
  \path[->] (qb0) edge              node [above left] {$a$}        (qb1)
                  edge              node [below left] {$a$}        (qb5)
            (qb1) edge              node [above]      {$b$}        (qb2)
            (qb2) edge              node [above]      {$a$}        (qb3)
            (qb3) edge              node [above]      {$b$}        (qb4)
            (qb4) edge [loop right] node [right]      {$a$}        ()
            (qb5) edge              node [above]      {$b$}        (qb6)
            (qb6) edge              node [above]      {$a$}        (qb7)
            (qb7) edge              node [above]      {$b$}        (qb8)
            (qb8) edge [loop right] node [right]      {$b$}        ();
\end{tikzpicture}}
\end{center}
\duplicator's winning strategy for the first three rounds is to store the three symbols chosen by \spoiler resulting in buffer contents
$(aa,b)$. Note that then the second buffer is full. Hence, when \spoiler chooses a $b$ now in state $q_3$, \duplicator must consume it in 
this round but in her current state $p_0$ she only has $a$-transitions available. Hence, she must consume the first $a$, advance to $p_1$
or $p_5$ depending on whether \spoiler has moved to $q_4$ or $q_5$. Say it was $q_4$ so \duplicator has moved to $p_1$ and the current
buffer content is $(a,bb)$. She then needs to move over to $p_2$ creating the buffer content $(a,b)$ and can end the round. She could also
consume the next buffered $a$ and move to $p_3$ with buffers $(\epsilon,b)$ but this does not help her, neither in this case nor in general.   
\end{example}

The hierarchy of buffered simulations stated in Proposition~\ref{prop:hierarchy} easily carries over to two-buffer simulations. The
total order $\le$ on $\Nat \cup \{\omega\}$ extends to the usual partial order of point-wise comparisons on pairs from
$\Nat \cup \{\omega\}$. I.e.\ we have $(78,6) \le (\omega,6)$ but $(3,4) \not\le (5,2)$.

\begin{theorem}
\label{thm:hierarchies}
For any $k_1,k_2,\ell_1,\ell_2 \in \Nat \cup \{\omega\}$ with $(k_1,k_2) \le (\ell_1,\ell_2)$ we have
$\simu{k_1,k_2}{} \subseteq \simu{\ell_1,\ell_2}{}$. Moreover, if $(\ell_1,\ell_2) \not\le (k_1,k_2)$ then there are automata $\mathcal{A}$ and $\mathcal{B}$ such that
$\mathcal{A} \notsimu{k_1,k_2}{} \mathcal{B}$ but $\mathcal{A} \simu{\ell_1,\ell_2}{} \mathcal{B}$.
\end{theorem}

\begin{proof}
We immediately get $\simu{k_1,k_2}{} \subseteq \simu{\ell_1,\ell_2}{}$
for $(k_1,k_2) \le (\ell_1,\ell_2)$ since any winning strategy for \duplicator in $\simugame{k_1,k_2}{\mathcal{A}}{\mathcal{B}}$
is also a winning strategy for her in $\simugame{\ell_1,\ell_2}{\mathcal{A}}{\mathcal{B}}$. Note that she is not required to use the additional
buffer capacities.

For the strictness part 
suppose w.l.o.g.\ that $\ell_1 > k_1$ which implies
$k_1 \in \Nat$. Then
 consider the two NBA $\mathcal{A}$ (left) and $\mathcal{B}$ (right) over $\Sigma$ as follows, and let
 $\Sigma_1= \{a \in \Sigma \mid \sigma(a) = 1\}$.
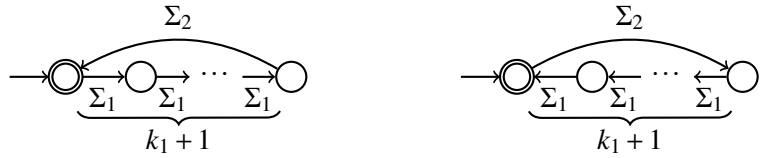
\begin{center}
\begin{tikzpicture}[initial text={}, node distance=10mm, every state/.style={minimum size=4mm}, thick]
  \node[state,initial,accepting]   (qa0)                      {};
  \node[state]                     (qa1) [right of=qa0]       {};  
  \node        		           (dot) [right of=qa1] 	    {$\cdots$};
  \node[state]                     (qa2) [right of=dot]       {};

  \path[->] (qa0) edge                 node [below]      {$\Sigma_1$}                       (qa1)
	    (qa1) edge 		       node [below]      {$\Sigma_1$} 	                    (dot)
	    (dot) edge                 node [below]      {$\Sigma_1$}                       (qa2)
	    (qa2) edge [bend right] node [above]      {$\Sigma_2$}     (qa0);
	    
  \node[state,initial,accepting]   (qb0)  [right of=qa2, node distance=3cm]   {};
  \node[state]                     (qb1)  [right of=qb0]                      {};  
  \node        		           (dot2) [right of=qb1]                      {$\cdots$};
  \node[state]                     (qb2)  [right of=dot2]                     {};

  \path[->] (qb1) edge                node [below]                        {$\Sigma_1$}     (qb0)
	    (dot2) edge               node [below]                        {$\Sigma_1$}     (qb1)
	    (qb2) edge                node [below]      {$\Sigma_1$}     (dot2)
	    (qb0) edge [bend left] node [above]      {$\Sigma_2$}     (qb2);
	    
\draw [decorate,decoration={brace,amplitude=5pt,mirror,raise=3mm}] 
(qa0.south east) -- (qa2.south west) node [midway, below=4mm] {${k_1}+1$};

\draw [decorate,decoration={brace,amplitude=5pt,mirror,raise=3mm}] 
(qb0.south east) -- (qb2.south west) node [midway, below=4mm] {$k_1+1$};
\end{tikzpicture}
\end{center}
It should be clear that we have $\mathcal{A} \simu{\ell_1,\ell_2} \mathcal{B}$ but
$\mathcal{A} \notsimu{k_1,k_2} \mathcal{B}$. 
\end{proof}

\subsection{Reductions to Ordinary Simulation}

The following theorem shows that bounded buffers can be eliminated at the cost of a blow-up in \duplicator's state space. In fact, the buffer can 
be incorporated
into \duplicator's automaton, resulting in an effective capacity of $0$ for this buffer. The structure needs to be defined such that \duplicator
can react immediately to any alphabet symbol that would have been put into that buffer, rather than consuming the older buffer content first and
thus advancing in the automaton to a different state.  

\begin{theorem}
  \label{thm:reduction}
  Let $k_1 \in \Nat \cup \{\omega\}$ and $k_2 \in \Nat$. For any
  automaton $\mathcal{B}$ of size $n$ there is an automaton $\mathcal{B}'$ of
  size $\leq 2n\cdot (|\Sigma|^{k_2+1}-1)$ such that for any automaton
  $\mathcal{A}$ we have that
  $\mathcal{A} \simu{k_1,k_2}{} \mathcal{B}$ iff
  $\mathcal{A} \simu{k_1,0} \mathcal{B}'$.
\end{theorem}

\begin{proof}
Intuitively, we save the content of the buffers in the state space of $\mathcal{B}'$.
Formally, let $\Delta :=\{a \in \Sigma \mid \sigma(a) = 2\}$ be the set of all letters that get stored in the buffer
of capacity $k_2$ which is to be reduced down to size $0$.
We write $\Delta^{\leq k}$ to denote $\{\epsilon \} \cup \Delta^1 \cup \ldots \cup \Delta^k$.
Let $\mathcal{B} = (Q,\Sigma,q_I,\delta,F)$. 
Then define $\mathcal{B}' := (Q \times \Delta^{\le k_2} \times \{0,1\}, \Sigma, (q_I,\epsilon,0), \delta', F \times \Delta^{\le k_2} \times \{0\})$
with 
\begin{align*}
\delta'((q,w,d),a) = \begin{cases}
\{ (q,aw,1) \} \cup \{ (p,v,0) \mid aw = vb, p \in \delta(q,b) \} &, \text{ if } a \in \Delta \text{ and } |w| < k_2, \\
\{ (p,av,0) \mid w = vb, p \in \delta(q,b) \} &, \text{ if } a \in \Delta \text{ and } |w|=k_2, \\
\{ (p,w,d) \mid p \in \delta(q,a) \} &, \text{ if } a \in \Sigma \setminus \Delta.
\end{cases}
\end{align*}
The third component in the states of $\mathcal{B}'$ is used to
indicate whether or not something has been put into the buffer.  It is
not difficult to transform \duplicator's winning strategies between
the games $\simugame{k_1,k_2}{\mathcal{A}}{\mathcal{B}}$ and
$\simugame{k_1,0}{\mathcal{A}}{\mathcal{B}'}$. Suppose
$\mathcal A\simu{k_1,k_2}\mathcal B$. Then, by \cite{GurH82},
\duplicator has a positional winning strategy $\sigma$ in the game
$\simugame{k_1,k_2}{\mathcal{A}}{\mathcal{B}}$.  Then we can define a
positional winning strategy $\sigma'$ for her in the latter via
$\sigma'(q,(\beta_1,\epsilon),(p,w,d)) := \sigma(q,(\beta_1,w),p)$ in
case of $|w|=k_2$. If $|w| < k_2$ then she simply follows the
deterministic transitions in $\mathcal{B}'$ which amount to waiting
for the buffer to be filled. The transformation of positional winning
strategies in the other direction works in the same way. Note that
there is no choice for her with respect to the value of $d$; its value
is determined by the value of the other components and the history of
a play.
\end{proof}

Clearly, the same argument can be used to reduce $(k_1,k_2)$-simulation to $(0,k_2)$-simulation when $k_1 \in \Nat$. 

The following result should be obvious given that two buffers of capacity $0$ do not introduce any partial commutativity between input symbols
since they always have to be consumed by \duplicator immediately after \spoiler has produced them. Hence, the order of consumption remains the
same as the order in which they are produced.

\begin{lemma}
\label{lem:simu0020}
$\simu{0,0} = \simu{0}$.
\end{lemma}

\subsection{Decidability and Complexity}

By applying Theorem~\ref{thm:reduction} twice we can transform the
input automata for a two-buffer simulation game with bounded buffers
into a pair of automata on which we need to check the relation
$\simu{0,0}{}$ which, by Lemma~\ref{lem:simu0020} and
Proposition~\ref{prop:hierarchy} (\ref{prop:hierarchy:ordsimu} is the
same as ordinary fair simulation. It is known \cite{EtessamiWS01} that
fair simulation games can be solved in polynomial time because they
are special cases of parity games of index 3, and parity games of
fixed index can be solved in polynomial time
\cite{EL86,Jurdzinski/00}.

\begin{corollary}
\label{cor:decidable_ptime}
For every fixed $k_1,k_2 \in \Nat$ , the relation $\simu{k_1,k_2}$ is decidable in polynomial time.
\end{corollary}

For unbounded buffers the situation is different. Decision problems involving two unbounded buffers typically become undecidable as stated in
the introduction. Here we adapt the argument used for the reachability problem for communicating 
finite state machines~\cite{BrandZ83,ChambartS08}, and show the
undecidability of $\simu{\omega,\omega}{}$ 
by a reduction from Post's Correspondence Problem~\cite{post46}.
The reduction basically constructs a pair of automata, such that
to win the game \spoiler is required to produce a possible solution for the PCP
and store it in the two unbounded buffers.
\duplicator's role is to check whether this is indeed a correct solution. 

\begin{theorem}
\label{undecidable}
The relation $\simu{\omega,\omega}{}$ is undecidable.
\end{theorem}

\begin{figure}[t]
\centering
\begin{tikzpicture}[initial text={}, node distance=12mm, every state/.style={minimum size=5mm}, thick]
  \node[state,initial]   (qa0)                                    {};
  \node                  (qa1)  [right of=qa0]                    {$\vdots$};
  \node[state]           (aux1) [right of=qa1]                    {};
  \node[state]           (qa2)  [above of=qa1, node distance=8mm] {};
  \node[state]           (qa3)  [below of=qa1, node distance=8mm] {};
  \node[state,accepting] (aux2) [right of=aux1]                   {};
  \node                  (A)    [above of=qa0]                    {$\mathcal{A}_L$};          

  \path[->] (qa0)  edge                    node [above left]      {$u_1$}       (qa2)
                   edge                    node [below left]      {$u_n$}       (qa3)
            (aux1) edge                    node [above]      {$\sharp \bar{\sharp}$}    (aux2)
            (aux2) edge [loop above]       node [above]      {$\sharp \bar{\sharp}$}    (aux2)
            (qa2)  edge                    node [above]      {$\bar{v}_1$} (aux1)
            (qa3)  edge                    node [below]      {$\bar{v}_n$} (aux1)      
            (aux1) edge [bend right = 60]  node [above right]      {$u_1$}       (qa2)
                   edge [bend left = 60]   node [below right]      {$u_n$}       (qa3);
            
  \node[state,initial]   (qb0) [right of=qa1, node distance=5cm]  {};
  \node[state]           (aux2)[right of=qb0, node distance=17mm] {};
  \node[state]           (qb1) [above of=aux2]                    {};
  \node[state]           (qb2) [above of=qb0]                     {};
  \node[state]           (qb3) [below of=aux2]                    {};
  \node[state]           (qb4) [below of=qb0]                     {};
  \node[state,accepting] (qb5) [right of=aux2, node distance=17mm]{};
  \node[]                (B)   [left of=qb2]                      {$\mathcal{B}_L$}; 

  \path[->] (qb0) edge [bend left]             node [left]       {$0$}               (qb2)
                  edge [bend right]            node [left]       {$1$}               (qb4)
            (qb1) edge                         node [above]      {$\bar{1},\bar{\sharp}$}  (qb5)
            (qb2) edge [bend left]             node [right]      {$\bar{0}$}         (qb0)
            (qb3) edge                         node [below]      {$\bar{0},\bar{1}$} (qb5)
            (aux2) edge                        node [above left] {$\bar{0},\bar{\sharp}$}  (qb5)
            (qb5) edge [loop above]            node [above]      {$\Sigma$}          (qb5)
            (qb4) edge [bend right]            node [right]      {$\bar{1}$}         (qb0)
            (qb0) edge                         node [above]      {$0$}               (qb1)
                  edge                         node [above right]{$1$}               (aux2)
                  edge                         node [below]      {$\sharp$}          (qb3);

\end{tikzpicture}
\caption{Automata $\mathcal{A}_L$ and $\mathcal{B}_L$ 
used in the proof of Theorem~\ref{undecidable}.}
\label{fig:autundecidable}
\end{figure}
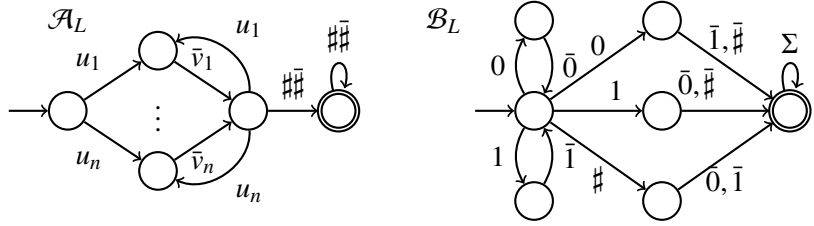

\begin{proof}
Let $L = \{(u_1,v_1), \ldots, (u_n,v_n)\}$ be an input for the PCP where $u_i$ and $v_i$ 
are non-empty finite words over $\{0,1\}$. We construct two automata $\mathcal{A}_L$
and $\mathcal{B}_L$ over $\Sigma = \{0,1,\sharp,\bar{0},\bar{1},\bar{\sharp}\}$.
We write $\bar{w}$ to denote $\bar{a}_1\ldots \bar{a}_m \in \{ \bar{0},\bar{1},\bar{\sharp} \}^*$ for $w = a_1\ldots a_m$.
We define $\sigma : \Sigma \rightarrow \{1,2\}$, where $\sigma(x) = 1$ for 
$x \in \{0,1,\sharp\}$, and $\sigma(x) = 2$ otherwise. Let $\mathcal{A}_L$
and $\mathcal{B}_L$ be the automata depicted in Figure~\ref{fig:autundecidable}.
We abbreviate a sequence of transitions with letters $a_1,\ldots,a_k$ in $\mathcal{A}_L$ with a single arrow.

An infinite word is accepted by $\mathcal{A}_L$ iff it is of the form 
$u_1\bar{v}_1\ldots u_n\bar{v}_n(\sharp \bar{\sharp})^\omega$. 
An infinite word is accepted by $\mathcal{B}_L$ 
if it starts with balanced pairs $(a\bar{a})^*$, $a\in \{0,1\}$, and
at one point contains an unbalanced pair $x\bar{y}$ with $x,y \in \{0,1,\sharp \}$ and $x \neq y$. 

We claim that
$\mathcal{A}_L \notsimu{\omega,\omega}{} \mathcal{B}_L$ iff there is
a solution for $L$. Suppose such a solution $i_1,\ldots,i_m$
exists. Then Spoiler wins by producing an accepting word
$u_{i_1}\bar{v}_{i_1}\ldots
u_{i_m}\bar{v}_{i_m}(\sharp\bar{\sharp})^\omega$. By
Lemma~\ref{lem:relationship}, \duplicator has to find a run on some word of the
form $x_1\bar{y_1} x_2\bar{y_2} x_3\bar{y_3}\dots$ with
$x_1x_2x_3\dots = u_{i_1} u_{i_2}\ldots
u_{i_m}\#^\omega=v_{i_1}v_{i_2}\ldots v_{i_m}\#^\omega
=y_1y_2y_3\dots$, i.e., without a mismatch. Hence she cannot win the
play.

On the other hand, if there is no solution for $L$, then no matter which path \spoiler chooses in his automaton it will
either not reach the accepting loop or it will have to contain a mismatch in the sense that the concatenation of the
$u$-parts and the concatenation of the $v$-parts differ in one digit at some position. \duplicator's winning strategy
consists of skipping her turn until such a mismatch is being produced which enables her to go to the accepting loop. 
Otherwise she waits forever but \spoiler does not produce an accepting run so she also wins such plays.
\end{proof}

\subsection{The Relationship to Language-Like Inclusion Problems}
\label{sub:lang_incl_prob}

Note that so far we have considered automata as finite-state devices with no particular semantics other than that provided by the
two-buffer games. Clearly, these automata are B\"uchi automata but we have avoided this analogy because two-buffer games, unlike
single-buffer games, lead away from the problem of language inclusion. We end this section on two-buffer games with a lemma that 
relates winning strategies in two-buffer games with a language-theoretic problem that boils down to language inclusion in the case
of a single buffer only. It is the analogy of case (\ref{prop:hierarchy:simimpliesli} of Proposition~\ref{prop:hierarchy} for
games with two buffers.

Suppose the alphabet $\Sigma$ and a distribution function $\sigma: \Sigma \to \{1,2\}$ is given. Let $\Sigma_i := \{ a \in \Sigma \mid \sigma(a) = i\}$ for
any $i \in \{1,2\}$. For a word $w \in \Sigma^\omega$ and an $i \in \{1,2\}$ we write $\project{w}{i}$ for the projection of $w$ onto $\Sigma_i$. 
Note that $\project{w}{i} \in \Sigma_i^* \cup \Sigma_i^\omega$.

Remember that $L(\mathcal{A})$ is used to denote the language of $\mathcal{A}$ seen as a B\"uchi automaton, i.e.\ the set of all words 
$w \in \Sigma^\omega$ for which there is an infinite path labelled by $w$ that visits final states infinitely often.

\begin{lemma}
\label{lem:relationship}
Let $\mathcal{A},\mathcal{B}$ be two automata over an alphabet $\Sigma$ with distribution function $\sigma: \Sigma \to \{1,2\}$. Let
$k_1,k_2 \in (\Nat \cup \{\omega\})$. If $\mathcal{A} \simu{k_1,k_2}{} \mathcal{B}$ then for every word $w \in L(\mathcal{A})$ there is
a $v \in L(\mathcal{B})$ such that for all $i\in\{1,2\}$ we have $\project{w}{i} = \project{v}{i}$.
\end{lemma}

\begin{proof}
Let $\mathcal{A} = (Q^{\mathcal{A}},\Sigma,q_I^{\mathcal{A}},\delta^{\mathcal{A}},F^{\mathcal{B}})$ and 
$\mathcal{B} = (Q^{\mathcal{A}},\Sigma,q_I^{\mathcal{A}},\delta^{\mathcal{A}},F^{\mathcal{B}})$. Suppose there is a word 
$w =a_0a_1\ldots \in L(\mathcal{A})$, i.e.\ there is an infinite path $\rho = q_0,a_0,q_1,a_1,\ldots$ such that $q_i \in F^{\mathcal{A}}$ for 
infinitely many $i$. Suppose furthermore that we have $\mathcal{A} \simu{k_1,k_2}{} \mathcal{B}$, i.e.\ player \duplicator has a winning 
strategy $\zeta$ for the game $\simugame{k_1,k_2}{\mathcal{A}}{\mathcal{B}}$. We remark that the values of $k_1$ and $k_2$ are irrelevant 
for what follows; only their existence is needed. 

So suppose that \spoiler chooses the run $\rho$. Following $\zeta$,
player $\duplicator$ will construct -- possibly in chunks -- an
infinite path $\rho' = p_0,b_0,p_1,b_1,\ldots$ through
$\mathcal{B}$. Since the resulting play is winning for \duplicator,
this path must contain infinitely many states in
$F^{\mathcal{B}}$. Thus, we have
$v := b_0b_1\ldots \in L(\mathcal{B})$. It remains to be seen that for
every $i \in \{1,2\}$ we have $\project{w}{i} = \project{v}{i}$. This
is a simple consequence of three facts.
\begin{enumerate}
\item \duplicator can only choose transitions with symbols that have been put into one of the buffers by \spoiler. Hence, for every $j$ 
      there is a $j'$ such that $b_j = a_{j'}$.  
\item Every symbol that gets put into the buffer is eventually removed from it. Hence,
      for every $j$ there is a $j'$ such that $a_j = b_{j'}$.
\item The buffers make sure that the order of two symbols from the same $\Sigma_i$ is being preserved.
\end{enumerate}
Thus, we get $\project{w}{i} = \project{v}{i}$ for every $i$. 
\end{proof}

Note that in the case of the single buffer simulation game, we necessarily have $v=w$ in the formulation
of this lemma which says nothing more than $L(\mathcal{A}) \subseteq L(\mathcal{B})$. In cases with two buffers, this lemma
predicts language inclusion of the two automata modulo partial commutativity between alphabet symbols of different $\sigma$-index.
The following section shows cases of problems in which such partial commutativity occurs naturally and shows how two-buffer games
can be used to characterise such problems.

\section{Application: Approximating Inclusion of $\omega$-Rational Relations}
\label{sec:transduce}

The main motivation for studying single-buffer games played on two B\"uchi automata is derived from the fact that such enhanced
simulations approximate language inclusion between B\"uchi automata, an important problem in the specification and verification of
reactive systems \cite{Emerson96}. Remember that simulations based on a single bounded buffer ($\simu{k}$ for some $k \in \Nat$)
provide polynomial-time approximations to a problem that is PSPACE-complete. Using an unbounded buffer defeats the purpose here because
$\simu{\omega}$ still only provides an approximation to language inclusion but additionally it is EXPTIME-complete 
\cite{DBLP:journals/corr/HutagalungLL14}, i.e.\ even harder than B\"uchi inclusion.

We consider the situation for rational relations\footnote{Here we only consider the case of binary relations. The generalisation to relations between
a larger number of words is straight-forward.}. Let $\Sigma_{\mathsf{in}}$ and $\Sigma_{\mathsf{out}}$ be finite alphabets, 
usually referred to as \emph{input} and \emph{output} alphabets. We write $\Sigma^\infty$ for $\Sigma^* \cup \Sigma^\omega$.
An \emph{infinitary rational relation} is an $R \subseteq \Sigma_{\mathsf{in}}^\infty \times \Sigma_{\mathsf{out}}^\infty$ that is recognised in 
the following sense.

\begin{definition}[see, e.g., \cite{GN84}]\rm
  A 2-head B\"uchi transducer is a $\mathcal{T} = (Q,\Sigma_{\mathsf{in}},\Sigma_{\mathsf{out}},q_I,\delta,F)$ where $Q$ is a finite set of states 
with a designated starting state $q_I \in Q$ and a designated set of \emph{final} states $F \subseteq Q$; $\Sigma_{\mathsf{in}}$ and $\Sigma_{\mathsf{out}}$ are two finite 
alphabets and $\delta \subseteq Q \times \Sigma_{\mathsf{in}}^* \times \Sigma_{\mathsf{out}}^* \times Q$ is a finite set of transitions.

A \emph{run} is an infinite sequence $q_0,u_0,v_0,q_1,u_1,v_1,\ldots$ over $Q \times \Sigma_{\mathsf{in}}^* \times \Sigma_{\mathsf{out}}^*$ such that for all $i \in \Nat$
we have $(q_i,u_i,v_i,q_{i+1}) \in \delta$. The run is accepting if $q_0 = q_I$ and there is some $q \in F$ such that $q = q_i$ for infinitely many $i$.
In that case, we say that the pair $(u,v) \in \Sigma_{\mathsf{in}}^\infty \times \Sigma_{\mathsf{out}}^\infty$ with $u = u_0u_1u_2\ldots$ and $v=v_0v_1v_2\ldots$ is
\emph{accepted} or \emph{recognised} by $\mathcal{T}$. 

The relation recognised by $\mathcal{T}$ is $R(\mathcal T) := \{ (u,v) \mid (u,v)$ is recognised by $\mathcal{T} \}$.
\end{definition}

The equivalence problem for infinitary rational relations -- given
$\mathcal{T}$, $\mathcal{T}'$, decide whether or not
$R(\mathcal{T}) = R(\mathcal{T}')$ -- and, hence, the inclusion
problem is undecidable~\cite{berstel1979transductions}, even for
deterministic transducers in the sense above. The latter can easily be
shown by a reduction from the infinitary PCP. 
Equivalence becomes decidable for B\"uchi transducers that read a pair of
exactly one input and one output symbol in each step and, hence, can
be seen as B\"uchi-automata over the alphabet
$\Sigma_{\mathsf{in}} \times \Sigma_{\mathsf{out}}$.

Transducers operating on infinite words have important applications. For instance, they are used to represent components of infinite structures, 
namely $\omega$-automatic ones \cite{LICS00*51} where decidability results are obtained for model-checking like problems; they
are used in concurrency to specify the synchronisation behaviour of parallel processes \cite{conf/caap/Nivat81}; they can represent
functions and relations on real numbers \cite{conf/csr/Carton10}; etc.

Given the undecidability of the inclusion problem for infinitary rational relations on one hand and their applications on the other, it
is fair to ask whether or not there are possibilities to approximate relation inclusion in the form of algorithms that are sound but
incomplete for instance. Such a possibility is given by the two-buffer simulations developed in the previous section: just like
single-buffer simulation games approximate language inclusion between B\"uchi automata, two-buffer games approximate inclusion between
infinitary rational relations, as is shown in the following.

\begin{definition} \rm
A 2-head B\"uchi transducer $\mathcal{T} = (Q,\Sigma_{\mathsf{in}},\Sigma_{\mathsf{out}},q_I,\delta,F)$ is \emph{normalised} if its transition relation is of the form
\begin{displaymath}
\delta \enspace \subseteq \enspace (Q \times \Sigma_{\mathsf{in}} \times \{ \epsilon \} \times Q) \cup (Q \times \{\epsilon\} \times \Sigma_{\mathsf{out}} \times Q)\ .
\end{displaymath}
\end{definition}

It should be clear that every 2-head B\"uchi transducer can be normalised preserving the relation that it recognises, and that this involves
a linear blow-up at most: every transition that consumes the input-output pair $(a_1\ldots a_n,b_1\ldots b_m)$ can be simulated by $n+m$
transitions, each of which consumes exactly one letter from either input or output. Moreover, w.l.o.g.\ we can assume 
$\Sigma_{\mathsf{in}} \cap \Sigma_{\mathsf{out}} = \emptyset$, i.e.\ every symbol is either input or output but not both. Then the transition relation of such a normalised
transducer can be seen as of type $Q \times (\Sigma_{\mathsf{in}} \cup \Sigma_{\mathsf{out}}) \times Q$ and, syntactically, this is nothing more than a B\"uchi automaton
as used in Section~\ref{sec:prel}.

\begin{theorem}
\label{thm:transducergame}
Let $\mathcal{T}$ and $\mathcal{T}'$ be two normalised 2-head B\"uchi transducers over inputs $\Sigma_{\mathsf{in}}$ and outputs $\Sigma_{\mathsf{out}}$
with $\Sigma_{\mathsf{in}} \cap \Sigma_{\mathsf{out}} = \emptyset$. Let $\Sigma := \Sigma_{\mathsf{in}} \cup \Sigma_{\mathsf{out}}$ with the alphabet
mapping function $\sigma(a) = 1$ if $a \in \Sigma_{\mathsf{in}}$ and $\sigma(a) = 2$ otherwise. Then we have: if 
$\mathcal{T} \simu{\omega,\omega} \mathcal{T}'$ then $R(\mathcal{T}) \subseteq R(\mathcal{T}')$. 
\end{theorem}

\begin{proof}
This follows immediately from Lemma~\ref{lem:relationship} with the observation that $R(\mathcal{T}) \subseteq R(\mathcal{T}')$ iff for every
word $w \in \Sigma^\omega$: if $(\project{w}{\mathsf{in}},\project{w}{\mathsf{out}}) \in R(\mathcal{T})$ then $(\project{w}{\mathsf{in}},\project{w}{\mathsf{out}}) \in R(\mathcal{T}')$.
\end{proof}

Combining this with Thm.~\ref{thm:hierarchies} and Cor.~\ref{cor:decidable_ptime} we obtain polynomial-time computable approximations for the
inclusion problem between infinitary recognisable relations.

\begin{corollary}
Let $k_1,k_2 \in \Nat$. We have that 
\begin{itemize}
\item it is decidable in polynomial time whether or not $\mathcal{T} \simu{k_1,k_2} \mathcal{T}'$ holds for arbitrary 2-head B\"uchi transducers, and
\item if $\mathcal{T} \simu{k_1,k_2} \mathcal{T}'$ holds then we have $R(\mathcal{T}) \subseteq R(\mathcal{T}')$.
\end{itemize}
\end{corollary}

The approximation is indeed not complete as the following example shows.

\begin{example}
Consider the normalised transducers $\mathcal{T}$ (left) and $\mathcal{T}'$ (right) over the alphabets $\Sigma_{\mathsf{in}} = \{ a\}$ and
$\Sigma_{\mathsf{out}} = \{ b,c\}$.
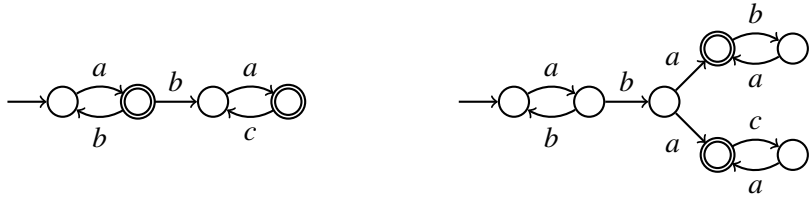
\begin{center}
\begin{tikzpicture}[initial text={}, node distance=10mm, every state/.style={minimum size=4mm}, thick]

  \node[state,initial]   (p0)                     {};
  \node[state,accepting] (p1) [right of=p0]       {};
  \node[state]           (p2) [right of=p1]       {};
  \node[state,accepting] (p3) [right of=p2]       {};

  \path[->] (p0) edge [bend left] node [above] {$a$} (p1)
            (p1) edge [bend left] node [below] {$b$} (p0)
                 edge             node [above] {$b$} (p2)
            (p2) edge [bend left] node [above] {$a$} (p3)
            (p3) edge [bend left] node [below] {$c$} (p2);

  \node[state,initial]   (q0) [right of=p3, node distance=3cm] {};
  \node[state]           (q1) [right of=q0]       {};
  \node[state]           (q2) [right of=q1]       {};
  \node[state,accepting] (q3) [above right of=q2] {};
  \node[state]           (q4) [right of=q3]       {};
  \node[state,accepting] (q5) [below right of=q2] {};
  \node[state]           (q6) [right of=q5]       {};

  \path[->] (q0) edge [bend left] node [above]      {$a$} (q1)
            (q1) edge [bend left] node [below]      {$b$} (q0)
                 edge             node [above]      {$b$} (q2)
            (q2) edge             node [above left] {$a$} (q3)
                 edge             node [below left] {$a$} (q5)
            (q3) edge [bend left] node [above]      {$b$} (q4)
            (q4) edge [bend left] node [below]      {$a$} (q3)
            (q5) edge [bend left] node [above]      {$c$} (q6)
            (q6) edge [bend left] node [below]      {$a$} (q5);
            
\end{tikzpicture}
\end{center}
Both recognise the infinitary relation described by $(a^\omega,b^\omega) \cup (a^\omega,b^*c^\omega)$. Thus, relation inclusion is given between
them. On the other hand, it is not difficult to see that \spoiler has a winning strategy for $\simugame{k_1,k_2}{\mathcal{T}}{\mathcal{T}'}$ for
all $k_1,k_2 \in \Nat \cup \{\omega\}$: he cycles on the left loop until \duplicator leaves her left cycle (which she has to eventually since it
contains no final states) and then, depending on where \duplicator went, either continues cycling there or moves over to his right cycle so 
that he produces symbols that \duplicator cannot consume anymore. 
\end{example}

\section{Conclusion and Further Work}
\label{sec:concl}

We have studied simulation games that use two FIFO buffers to store
\spoiler's choices for a limited amount of time before \duplicator has
to respond to them. We have shown how the correspondence between
(single-)buffered simulation games and the language inclusion problem
for B\"uchi automata naturally extends to two-buffer simulations and
inclusion problems for rational relations over infinite words. The
fact that games with bounded capacities can be solved in polynomial
time, and that higher buffer capacities do not make \duplicator weaker
yields polynomial-time approximations to such inclusion problems with
underlying partial commutativity.

Our results are closely related to the theory of (infinite)
Mazurkiewicz traces (see \cite{dm97handbook} for a general overview on
this theory). The reason is that relations in
$\Sigma_1^\omega\times\Sigma_2^\omega$ can be seen as languages of
real traces over particular independence alphabets (namely, complete
bipartite ones). More precisely: Fix a mapping
$\sigma\colon\Sigma\to\{1,2\}$ and let
$\Sigma_i=\{a\in\Sigma\mid \sigma(a)=i\}$. Furthermore, let
$\mathcal A$ and $\mathcal B$ be two B\"uchi-automata with
$\mathcal A\simu{k_1,k_2}\mathcal B$. Then, for any word
$w\in L(\mathcal A)$, there exists a word $v\in L(\mathcal B)$ with
$\project{w}{i}=\project{v}{i}$ for $i=1,2$. In trace-theoretic terms,
this means that the trace closure of $L(\mathcal A)$ is contained in
the trace closure of $L(\mathcal B)$ (where the underlying
independence relation is
$I=(\Sigma_1\times\Sigma_2)\cup(\Sigma_2\times\Sigma_1)$, i.e.,
complete bipartite). Hence, for such special independence relations,
we obtain polynomial-time computable approximations for the inclusion
problem of trace closures of $\omega$-regular languages (which is
undecidable). It is possible to extend the results of this paper
to $m\ge 2$ buffers, which allows similar polynomial-time
computable approximations to be obtained in case of complete $m$-partite independence
relations. In our ongoing work \cite{HutHKLL16},
we extend $\sigma$ and the
multi-buffer simulation game in a way that allows arbitrary
independence relations to be captured (not just complete $m$-partite ones).

Regarding decidability,
we show that
$\simu{\omega,0}{}$ is highly undecidable (i.e., not
arithmetical) \cite{HutHKLL16}. This is in sharp contrast with the decidability result
for $\simu{\omega}{}$ from~\cite{DBLP:journals/corr/HutagalungLL14}
and the fact that bounded buffers can be encoded in the control state.

An idea for further work is the following.  Recall that
$\simu{\omega}{}$ is EXPTIME-complete and it has a PSPACE-complete
variant in which \duplicator may never consume an element from the
buffer \emph{and} leave others in, i.e.\ she always has to flush the
buffer whenever she moves
\cite{DBLP:journals/corr/HutagalungLL14}. One can study such a variant
for two-buffer simulations as well. This becomes technically a little
bit more tedious but not conceptually problematic; one could require
her to flush \emph{both} buffers or just \emph{one}, etc. The same
reduction of PCP shows that the simulation remains undecidable if two
buffers are unbounded. However, it would be interesting to determine
the status of $\simu{\omega,k}{}$: is it highly undecidable as
$\simu{\omega,0}$ \cite{HutHKLL16},
arithmetical or even decidable
and, if so, what is its precise complexity.

\bibliographystyle{eptcs}
\bibliography{literature}

\end{document}